\documentclass[conference]{IEEEtran}
\IEEEoverridecommandlockouts

\usepackage[acronym,nonumberlist]{glossaries}
\glsdisablehyper
\usepackage[ruled]{algorithm}
\usepackage[noend]{algpseudocode}
\usepackage{algorithmicx}
\usepackage{comment}
\usepackage{graphicx}
\usepackage{multirow}
\usepackage{amsmath,amsfonts,amssymb}
\usepackage{amsthm}
\usepackage{mathtools}
\usepackage{cite}
\usepackage{url}
\usepackage{xcolor}
\usepackage{hyperref}
\usepackage{booktabs}
\usepackage{makecell}
\def\BibTeX{{\rm B\kern-.05em{\sc i\kern-.025em b}\kern-.08em
    T\kern-.1667em\lower.7ex\hbox{E}\kern-.125emX}}

\newtheorem{theorem}{Theorem}

\newtheorem{definition}{Definition}
\newtheorem{remark}{Remark} % Unnumbered remark
\newacronym{}{}{}

\begin{document}

\newacronym{ACP}{ACP}{Accelerator Coherency Port}
\newacronym{ADC}{ADC}{Analog-to-Digital Converter}
\newacronym{ADAS}{ADAS}{Advanced Driver Assistance System}
\newacronym{ADESD}{ADESD}{Application-Driven Embedded System Design}
\newacronym{ADPCM}{ADPCM}{Adaptive Differential Pulse-Code Modulation}
\newacronym{AES}{AES}{Advanced Encryption Standard}
\newacronym{AET}{AET}{Average Execution Time}
\newacronym{AI}{AI}{Artificial Intelligence}
\newacronym{ALU}{ALU}{Arithmetic Logic Unit}
\newacronym{AMBA}{AMBA}{ARM Advanced Microcontroller Bus Architecture}
\newacronym{ANN}{ANN}{Artificial Neural Network}
\newacronym{AOP}{AOP}{Aspect Oriented Programming}
\newacronym{API}{API}{Application Program Interface}
\newacronym{ASIC}{ASIC}{Application-Specific Integrated Circuit}
\newacronym{AV}{AV}{Autonomous Vehicle}
\newacronym{AXI}{AXI}{Advanced eXtensible Interface}
\newacronym{BLEU}{BLEU}{Bilingual Evaluation Understudy}
\newacronym{CA}{CA}{Certificate Authority}
\newacronym{CAM}{CAM}{Cooperative Awareness Message}
\newacronym{CAP-OS}{CAP-OS}{Configuration Access Port Operating System}
\newacronym{CI}{CI}{Confidence Interval}
\newacronym{CPM}{CPM}{Collective Perception Message}
\newacronym{CPN}{CPN}{Counter Propagation Network}
\newacronym{CPS}{CPS}{Cyber-Physical System}
\newacronym{CV}{CV}{Computer Vision}
\newacronym{DAC}{DAC}{Digital-to-Analog Converter}
\newacronym{DAG}{DAG}{Directed Acyclic Graph}
\newacronym{DDR}{DDR}{Double Data Rate}
\newacronym{DDS}{DDS}{Data Distribution Service}
\newacronym{DENM}{DENM}{Decentralized Environmental Notification Message}
\newacronym{DFT}{DFT}{Discrete Fourier Transform}
\newacronym{DMA}{DMA}{Direct Memory Access}
\newacronym{DPM}{DPM}{Dynamic Power Management}
\newacronym{DSP}{DSP}{Digital Signal Processing}
\newacronym{DTMF}{DTMF}{Dual-Tone Multi-Frequency}
\newacronym{DVFS}{DVFS}{Dynamic Voltage-Frequency Scaling}
\newacronym{ECC}{ECC}{Elliptic Curve Cryptography}
\newacronym{ECDH}{ECDH}{Elliptic Curve Diffie-Hellman}
\newacronym{ECDSA}{ECDSA}{Elliptic Curve Digital Signature Algorithm}
\newacronym{ECU}{ECU}{Electronic Control Unit}
\newacronym{EDA}{EDA}{Electronic Design Automation}
\newacronym{ED}{ED}{Event-Driven}
\newacronym{EDF}{EDF}{Earliest Deadline First}
\newacronym{ELSF}{ELSF}{Optimal Effective Least Slack First}
\newacronym{EPOS}{EPOS}{Embedded Parallel Operating System}
\newacronym{ESA}{ESA}{External Security Agent}
\newacronym{ESM}{ESM}{Erlangen Slot Machine}
\newacronym{ETSI}{ETSI}{European Telecommunications Standards Institute}
\newacronym{FANN}{FANN}{Fast Artificial Neural Network}
\newacronym{FBD}{FBD}{Family-Based Design}
\newacronym{FFT}{FFT}{Fast Fourier Transform Transform}
\newacronym{FoV}{FoV}{Field-of-View}
\newacronym{FPGA}{FPGA}{Field-Programmable Gate Array}
\newacronym{FPP}{FPP}{Fast Passive Parallel}
\newacronym{FTTSTP}{FT-TSTP}{Fault-Tolerant Trustful Space-Time Protocol}
\newacronym{GIP}{GIP}{Gateway Integrity Checking Protocol}
\newacronym{GPS}{GPS}{Global Positioning System}
\newacronym{GPU}{GPU}{Graphics Processing Unit}
\newacronym{HDL}{HDL}{Hardware description language}
\newacronym{HDV}{HDV}{Heavy-duty Vehicle}
\newacronym{HECOPS}{HeCoPS}{Heuristic Environmental Consideration Over Positioning System}
\newacronym{HLS}{HLS}{High-Level Synthesis}
\newacronym{HPC}{HPC}{Hardware Performance Counter}
\newacronym{HTI}{HTI}{Hardware Thread Interface}
\newacronym{ICAP}{ICAP}{Internal Configuration Access Port}
\newacronym{iLBC}{iLBC}{Internet Low Bitrate Codec}
\newacronym{IoT}{IoT}{Internet of Things}
\newacronym{IMU}{IMU}{Inertial Measurement Unit}
\newacronym{ISA}{ISA}{Instruction Set Architecture}
\newacronym{ISR}{ISR}{Interrupt Service Routine}
\newacronym{ITS}{ITS}{Intelligent Transportation System}
\newacronym{LAN}{LAN}{Local Area Network}
\newacronym{LLF}{LLF}{Least Laxity First}
\newacronym{LRG}{LRG}{Least Recently Granted}
\newacronym{LSF}{LSF}{Least Slack First}
\newacronym{MAC}{MAC}{Dynamic Random Access Memory}
\newacronym{MAPE}{MAPE}{Mean Absolute Percent Error}
\newacronym{ML}{ML}{Machine Learning}
\newacronym{MPI}{MPI}{Message Passing Interface}
\newacronym{MPSoC}{MPSoC}{Multiprocessor System-on-Chip}
\newacronym{MV}{MV}{Motion Vector}
\newacronym{NBTI}{NBTI}{Negative-Bias Temperature Instability}
\newacronym{NLP}{NLP}{Natural Language Processing}
\newacronym{NN}{NN}{Neural Network}
\newacronym{NoC}{NoC}{Network on Chip}
\newacronym{NT}{NT}{Neural Translation}
\newacronym{OBU}{OBU}{On-board Unit}
\newacronym{OL-EAMS}{OL-EAMRTS}{Online-Learning Energy-Aware Multicore Real-Time Scheduler}
\newacronym{OOP}{OOP}{Object Oriented Programming}
\newacronym{OPCUA}{OPC UA}{OPC Unified Architecture}
\newacronym{OSIF}{OSIF}{Operating System Interface}
\newacronym{OS}{OS}{Operating System}
\newacronym{OTP}{OTP}{One-Time Password}
\newacronym{PABX}{PABX}{Private Automatic Branch Exchange}
\newacronym{PCA}{PCA}{Principal Component Analysis}
\newacronym{PCAP}{PCAP}{Processor Configuration Access Port}
\newacronym{PCM}{PCM}{Pulse-Code Modulation}
\newacronym{PIO}{PIO}{Programmed Input-Output}
\newacronym{PLC}{PLC}{Programmable Logic Controller}
\newacronym{PMU}{PMU}{Performance Monitoring Unit}
%\newacronym{PMU}{PMU}{Power Management Unit}
\newacronym{PTP}{PTP}{Precision Time Protocol}
\newacronym{PVT}{PVT}{Process, Voltage and Temperature}
\newacronym{QoS}{QoS}{Quality of Service}
\newacronym{RAM}{RAM}{Random-Access Memory}
\newacronym{RFE}{RFE}{Recursive Feature Elimination}
\newacronym{RI}{RI}{Reconfiguration Interface}
\newacronym{RL}{RL}{Reinforcement Learning}
\newacronym{RMI}{RMI}{Remote Method Invocation}
\newacronym{RM}{RM}{Rate-Monotonic}
\newacronym{RMSD}{RMSD}{Root-Mean-Square Deviation}
\newacronym{RSA}{RSA}{Rivest-Shamir-Adleman Public-key Cryptography}
\newacronym{RSS}{RSS}{Responsibility-Sensitive Safety}
\newacronym{RSU}{RSU}{Roadside Unit}
%\newacronym{RTAMT}{RTAMT}{Run-time Analog Monitoring Tool}
\newacronym{RTAMT}{RTAMT}{Run-time Assurance Monitoring Tool}
\newacronym{RTL}{RTL}{Register-Transfer Level}
\newacronym{RTOS}{RTOS}{Real-Time Operating System}
\newacronym{RTSNoC}{RTSNoC}{Real-Time Star Network-on-Chip}
\newacronym{SEU}{SEU}{Safety Enforcement Unit}
\newacronym{SIMD}{SIMD}{Single Instruction, Multiple Data}
\newacronym{SI}{SI}{International System of Units}
\newacronym{SoC}{SoC}{System-on-a-Chip}
\newacronym{SPTP}{SPTP}{Speculative Precision Time Protocol}
\newacronym{SRAM}{SRAM}{Static Random-Access Memory}
\newacronym{SSB}{SSB}{Smart Solar Building}
\newacronym{STL}{STL}{Signal Temporal Logic}
\newacronym{TA}{TA}{Trusted Authority}
\newacronym{TDMA}{TDMA}{Time Division Multiple Access}
\newacronym{TL}{TL}{Temporal Logic}
\newacronym{TSCH}{TSCH}{Time-Synchronized Channel Hopping}
\newacronym{TSN}{TSN}{Time-Sensitive Networking}
\newacronym{TSSDN}{TSSDN}{Time-Sensitive Software-Defined Networks}
\newacronym{TSTP}{TSTP}{Trustful Space-Time Protocol}
\newacronym{TT}{TT}{Time-Triggered}
\newacronym{V2I}{V2I}{Vehicle-to-Infrastructure Communication}
\newacronym{V2V}{V2V}{Vehicle-to-Vehicle Communication}
\newacronym{VHDL}{VHDL}{VHSIC Hardware Description Language}
\newacronym{VLAN}{VLAN}{Virtual Local Area Networking}
\newacronym{WCET}{WCET}{Worst Case-Execution Time}
\newacronym{WCRT}{WCRT}{Worst-Case Response Time}
\newacronym{WSN}{WSN}{Wireless Sensor Network}

% Acronyms added from Murillo's undergrad thesis:

\newacronym{PINN}{PINN}{Physics-Informed Neural Network}
\newacronym{CARLA}{CARLA}{Cars Learning to Act}
\newacronym{FCW}{FCW}{Forward Collision Warning}
\newacronym{AEB}{AEB}{Automatic Emergency Braking}
\newacronym{ASIL}{ASIL}{Automotive Safety Integrity Levels}
\newacronym{FNN}{FNN}{Feedforward Neural Network}
\newacronym{RNN}{RNN}{Recurrent Neural Network}
\newacronym{LSTM}{LSTM}{Long Short-Term Memory}
\newacronym{RLS}{RLS}{Recursive Least Squares}
\newacronym{LSM}{LSM}{Least Mean Squares}
\newacronym{VLDM}{VLDM}{Vehicle Longitudinal Dynamical Model}
\newacronym{FuSa}{FuSa}{Functional Safety}
\newacronym{SLSD}{SLSD}{Same Lane Same Direction}
\newacronym{UE4}{UE4}{Unreal Engine 4}
\newacronym{SD}{SD}{SmartData}
\newacronym{SDAV}{SDAV}{SmartData for Autonomous Vehicles}
\newacronym{DL}{DL}{Deep Learning}
\newacronym{MB}{MB}{Model-Based}
\newacronym{PDE}{PDE}{Partial Differential Equations}
\newacronym{RPM}{RPM}{Rotations Per Minute}

\newacronym{CAN}{CAN}{Control Area Network}

\title{Ensuring Data Freshness in Multi-Rate Task Chains Scheduling}

\author{\IEEEauthorblockN{1\textsuperscript{st} José Luis Conradi Hoffmann \qquad 2\textsuperscript{nd} Antônio Augusto Fröhlich}
\IEEEauthorblockA{\textit{Software/Hardware Integration Lab, Federal University of Santa Catarina} \\
\textit{Florianópolis, Brazil} \\
\{hoffmann, guto\}@lisha.ufsc.br}
}

\maketitle

\begin{abstract}
In safety-critical autonomous systems, data freshness presents a fundamental design challenge. While the Logical Execution Time (LET) paradigm ensures compositional determinism, it often does so at the cost of injected latency, possibly degrading the age of data on high-frequency control loops. Furthermore, heterogeneous, multi-rate, task dependencies is typically guaranteed inefficiently through oversampling. This paper proposes a Task-based scheduling framework extended with data freshness constraints. Unlike traditional models, scheduling decisions are driven by the lifespan of data. We introduce a formal methodology to decompose Data Dependency Graphs into dominant paths by tracing the strictest data freshness constraints backward from the actuators. Based on this decomposition, we propose an offset search algorithm that synchronizes multi-rate, multi-dependencies, task chains. This approach enforces end-to-end data freshness without the artificial latency of LET buffering, a trade-off between data freshness and execution determinism. We formally prove that this offset-based alignment preserves the 100\% schedulability capacity of Global EDF while addressing data freshness guarantees.
\end{abstract}


\section{Introduction}\label{sec:intro}

The paradigm shift towards high-level autonomous driving demands software architectures capable of managing distributed critical applications while guaranteeing safety constraints. Modern automotive systems rely on distributed multi-core Electronic Control Units (ECUs) interfaced by standards such as AUTOSAR to abstract functional logic from hardware topology \cite{autosar2022timing}. The correctness of these systems depends not only on the logical result of computation but also on the time at which these results are produced \cite{buttazzo2011hard}. The foundational \textit{Time-Triggered Architecture (TTA)} by Kopetz established that a global time base is a prerequisite for composability in distributed real-time systems, simplifying error detection and guaranteeing timeliness \cite{Kopetz:TTA-IEEE:2003}. Today, these concepts underpin safety certification under ISO 26262, which relies on the Fault Tolerant Time Interval (FTTI), the minimum time span between a fault and a possible occurrence of a hazardous event \cite{iso26262} if a safe mechanims is not activated. Consequently, data age becomes as critical as the logical correctness of the control algorithms themselves.

To address temporal interference, the Logical Execution Time (LET) paradigm has emerged as a dominant design pattern \cite{henzinger2003giotto, kirsch2012let}. By decoupling communication from execution, LET ensures that task input/output timing is independent of execution platform speed. Recent advancements have extended this to distributed networks via System-Level LET (SL-LET) \cite{gemlau2022system}, and proposed algebraic approaches to collapse task chains into zero-jitter logical tasks \cite{bini2023zero}. However, achieving this determinism often forces a trade-off. To guarantee that end-to-end latency remains within the FTTI while adhering to fixed LET intervals, designers typically over-dimension resources or accept stale data, effectively trading data age for scheduling predictability. For fast dynamic systems, such as Lane Keeping Assist (LKA), using the freshest data is often more critical for maintaining stability than eliminating jitter. If lateral position data is old due to buffering, the steering controller calculates a reaction for a road curvature the vehicle has already passed, potentially leading to safety breaches.

To address this trade-off, we propose a fundamental inversion of the design model: rather than permitting the scheduler to execute tasks as soon as possible, release times are constrained by the data freshness implied by task dependency. Unlike traditional real-time analysis where tight deadlines dictate high priority, we argue that tasks with large data freshness windows (higher laxity) should be assigned earlier offsets. This ensures that high-latency operations are processed in advance, effectively anchoring the timeline. Conversely, tasks with strict data freshness constraints (short data freshness) are assigned later offsets, delaying their execution until the precise moment their data are required.

In this work, we formalize the logic to synchronize heterogeneous sensor-to-actuator task chains. By analyzing the task dependency graph backwards from the actuation sink, we establish release offsets based on data freshness. This prefetching of high-laxity data (e.g., camera images) ensures that highly dynamic data (e.g., IMU) are sampled Just-In-Time (JIT) to match the laxity window of high-laxity ones. Once the offsets are applied, a runtime scheduler can manage execution using standard Earliest Deadline First (EDF). This hybrid approach eliminates the \textbf{stale data} problem caused by uncoordinated early execution without the artificial latency overhead of Logical Execution Time (LET) buffering.

\subsection{Contributions}The main contributions of this paper are summarized as follows:
\begin{enumerate}
    \item \textbf{Formal Path Decomposition:} A methodology for decomposing Task Dependency DAGs into dominant data freshness paths, identifying the temporal relationship between high-latency tasks and highly dynamic data.
    \item \textbf{Offset Assignment based on Data Freshness:} A start-time optimization algorithm that inverts standard urgency metrics, assigning earlier static offsets to tasks with larger data freshness to anchor the fusion window.
    \item \textbf{Shared Producer Consensus Search:} A backtracking algorithm to resolve offset conflicts in shared producer tasks, ensuring data freshness across heterogeneous multi-rate consumer chains without breaking periodicity.
    \item \textbf{Schedulability Preservation Proof:} A formal proof demonstrating that the introduction of data freshness-aware offsets preserves the $100\%$ capacity of Global EDF, ensuring that data freshness is achieved without degrading the system's processor utilization bounds.
\end{enumerate}

The remainder of this paper is organized as follows: Section \ref{sec:related} presents this work position in regards to other works from the state of the art addressing multi-task chain. Section \ref{sec:model} presents the definitions pertaining the system model. Section \ref{sec:proposal} presents the proposed solution to adapt start offset of tasks to optimize data freshness in task chains. Section \ref{sec:proof} presents the formal proof of the schedulability capacity preservation of Global EDF. Finally, Section \ref{sec:conclusion} concludes the paper.

\section{Related Works}\label{sec:related}

The scheduling of safety-critical automotive systems has traditionally been approached through the lens of task-centric deadlines. This section reviews the evolution from classical hard real-time scheduling to modern data-centric paradigms, highlighting the shift from task schedulability to end-to-end data freshness in Directed Acyclic Graphs (DAGs).

\subsection{Predictability in Task Chains}
The foundation of hard real-time systems lies in guaranteeing that critical tasks meet their deadlines under worst-case conditions. Fixed Priority Scheduling established rigorous Response Time Analysis (RTA) for single-core and distributed systems \cite{tindell1994holistic}. 
To address composability, Kopetz introduced the Time-Triggered Architecture (TTA) \cite{Kopetz:TTA-IEEE:2003}, advocating for a global time base to decouple communication from processing.

However, modern autonomous systems have evolved into complex, multi-rate Cause-Effect Chains. Becker \textit{et al.} \cite{Becker:2017} proposed a comprehensive end-to-end timing analysis for such chains, enabling analysis based on \gls{WCET} estimates even without detailed platform knowledge. While effective for verifying register-based communication, this model faces challenges in modern architectures where hidden latencies introduced by security algorithms, transient network loads, and complex AI data dependencies significantly impact system performance.

Recent research into Conditional DAGs has addressed the complexity of intra-task priority assignment to minimize response times on multiprocessors \cite{Qingqiang2023} and Response-Time Analysis for conditional branches in tasks \cite{melani2015}. While these methods optimize for \textit{execution urgency}, they do not inherently account for \textit{data freshness}.

\subsection{Logical Execution Time (LET) and Task Chains}
To reconcile preemption flexibility with temporal determinism, the Logical Execution Time (LET) paradigm \cite{henzinger2003giotto} fixes task input/output timing. This concept, standardized in AUTOSAR \cite{autosar2022timing}, has been the subject of intense optimization. 

Bini \textit{et al.} \cite{bini2023zero} proposed an algebraic framework to collapse periodic LET chains into single zero-jitter logical tasks. While this eliminates jitter, it enforces rigid structural latencies. Addressing multi-core contention, Wang \textit{et al.} \cite{bini2025jitter} highlights that uncontrolled jitter is a major threat to data integrity, providing a linear-time complexity analysis on jitter at read and write events. Our solution uses a Consensus Search to fix static offsets, effectively damping the jitter propagation analyzed in their work. The data freshness constraint is used to build a schedule where the temporal slack is strategically placed at the start, making the system naturally robust to the propagation effects.

\subsection{Task Dependency and Data Freshness}
Palencia et al. \cite{palencia1998schedulability} established that offsets can be used to mitigate the pessimism of Response Time Analysis (RTA) in transactions, their approach remains task-centric, focusing on processor interference. Our work extends the utility of offsets into the semantic domain, using them to align the data freshness windows of heterogeneous data streams. We demonstrate that by construction, specifically through the offset assignment based on data freshness, we can ensure data freshness while maintaining the capacity preservation.

Recent literature has increasingly focused on semantic metrics of data freshness, distinguishing between simple precedence constraints and actual temporal properties like \textit{Data Age} and \textit{Reaction Time}. Gohari \textit{et al.} \cite{Gohari2022} has provided a rigorous analysis for lower and upper bound data age in multi-rate DAGs under timing uncertainty. However, their solution remains primarily analytical, quantifying the data freshness degradation inherent in existing schedules. 
Günzel et al. \cite{Guenzel2023} highlights the increasing complexity of automotive timing analysis, where heterogeneous communication mechanisms (e.g., LET, Implicit Communication) introduce varying levels of latency into cause-effect chains.

Existing frameworks for DAGs primarily focus on bounding the Worst-Case Data Age (WCDA) under standard scheduling policies where tasks are released as soon as possible. In these models, data age accumulates passively while fresh data waits in output buffers for consumer activation. While these works provide rigorous \textit{analysis} methods to bound the age of a given schedule, they generally assume the schedule is driven by standard execution priorities (e.g., EDF or Rate Monotonic). Our work differs by shifting the focus from \textit{analysis} to \textit{construction}. Rather than accepting the data freshness degradation inherent in regular scheduling, we employ a \textit{Start-Time Optimization} strategy to actively inject delays, ensuring data freshness alignment by design. 

While frameworks like Zero-Jitter LET \cite{bini2023zero} achieve synchronization by buffering (delaying) the \textit{output}, our method achieves synchronization by delaying the \textit{input} (the sampling instant) of the faster tasks. This allows the system to handle high-latency dependencies (like computer vision tasks) by starting them early, while the critical path (dominant dependency) is scheduled to execute only when the high-latency data is nearing availability, thereby ensuring data freshness at the moment of fusion.

To illustrate the difference, consider a fusion task $\tau_{fus}$ that requires data from a high-latency computer-vision task ($\tau_{vis}, C_{vis}=10ms$), where $C_i$ is the Worst-Case Estimation Time (WCET), and a low-latency IMU processing task ($\tau_{imu}, C_{imu}=1ms$). We assume two parallel cores:
\begin{itemize}
    \item \textbf{Standard Approach (Analysis):} Both tasks are released at $t=0$. $\tau_{imu}$ finishes at $t=1$, while $\tau_{cam}$ finishes at $t=10$. The fusion occurs at $t=10$. Consequently, the IMU data sits in the output buffer for $\Delta t = 9ms$, resulting in a data age of $10ms$ at the point of fusion.
    \item \textbf{Proposed Approach (Construction):} We apply a \textit{Start Offset} to the dominant task. $\tau_{cam}$ starts at $t=0$ (due to its large data freshness window/latency). However, $\tau_{imu}$ is offset to start at $t=9$. Both tasks finish exactly at $t=10$. The IMU data age at fusion is minimal ($1ms$), eliminating the $9ms$ of "stale" waiting time.
\end{itemize}
Our approach shifts from analyzing the age of execution flows to constructing flows that minimize waiting times by design.

\section{System and Task Model}\label{sec:model}

We consider a real-time application modeled as a Directed Acyclic Graph (DAG) $G = (\mathcal{T}, \mathcal{E})$, where nodes represent computational tasks and edges represent data dependencies.

\subsection{Task Model}
The task set $\mathcal{T} = \{ \tau_1, \dots, \tau_n \}$ consists of $n$ periodic tasks. Each task $\tau_i$ is characterized by the tuple:
\begin{equation}
    \tau_i = (C_i, T_i, \Phi_i, P_i)
\end{equation}
where:
\begin{itemize}
    \item $C_i$ denotes the \textbf{Worst-Case Execution Time}.
    \item $T_i$ represents the \textbf{Period} of the task.
    \item $P_i$ is the priority of the task.
    \item $\Phi_i$ is the \textbf{Start Offset} relative to the period arrival. In this paper, $\Phi_i$ is a synthesis variable used to enforce data freshness alignment.\footnote{While Bini et al.~\cite{bini2023zero} also use offsets, they apply them to delay the \textit{output} (buffering) to force regularity. In contrast, we apply $\Phi_i$ to delay the \textit{input} (release) to minimize waiting time.}
\end{itemize}

The set of edges $\mathcal{E} \subseteq \mathcal{T} \times \mathcal{T}$ defines the flow of information. Each directed edge $e_{i,j} = (\tau_i, \tau_j) \in \mathcal{E}$ represents a dependency where $\tau_i$ produces data consumed by $\tau_j$, associated with a specific data freshness constraint $E_{i,j}$.
This parameter defines the maximum relative duration that data produced by $\tau_i$ remains useful for a consumer $\tau_j$. Otherwise we say the data has expired, and the computational validity becomes 0. For instance, in a Least Laxity First Schedule, the priority would be equal to the data freshness constraint $E_{i,j}$.

This concept is based on \cite{Frohlich:SD-IJSN:2018}, and it can be derived from the physical dynamics, such as the ability of the data to still represent the system environment, based on the control loop definitions. When compared to deadline, the expiry constraint is relative to the completion of the task and dictates the limit for its consumption by the depending task, while the deadline constraint is relative to the release of the task and dictates the limit for the task itself to finish.

To rigorously define data freshness, we establish the following temporal reference points for any $k$-th job of task $\tau_i$:
\begin{definition}[Release Time $r_{i,k}$]
The instant the $k$-th job of task $i$ is instantiated and becomes ready for scheduling plus the specified offset.
\begin{equation}
    r_{i,k} = (k-1)T_i + \Phi_i
\end{equation}
\end{definition}
\begin{definition}[Start Time $s_{i,k}$]
The instant the job $k$ of task $i$ actually begins execution on the processor. For a consumer task, this is the moment it reads its input buffers.
\begin{equation}
    s_{i,k} \geq r_{i,k}
\end{equation}
\end{definition}
\begin{definition}[Finish Time $f_{i,k}$]
The instant the $k$-th job of task $i$ completes its execution and produces output data.
\begin{equation}
    f_{i,k} \leq s_{i,k} + C_i
\end{equation}
\end{definition}

From the set of Edges $\mathcal{E} \in G$, we define:
\begin{definition} \textbf{Successor Set ($succ(\tau_i)$):} The set of tasks that are dependent on task $\tau_i$ (e.g., consumers).
    \begin{equation}
        succ(\tau_i) = \{ \tau_j \in \mathcal{T} \mid (\tau, \tau_j) \in \mathcal{E} \}
    \end{equation}
\end{definition}
\begin{definition} \textbf{Predecessor Set ($pred(\tau)$):} The set of dependencies of task $\tau_i$ (e.g., producers).
    \begin{equation}
        pred(\tau_i) = \{ \tau_j \in \mathcal{T} \mid (\tau_j, \tau_i) \in \mathcal{E} \}
    \end{equation}    
\end{definition}

Finally, $J_i$ is the set of jobs of task $\tau_i$.

\subsubsection{System Assumptions}
We assume the following data dependency semantics for the remaining of the paper:
\begin{itemize}
    \item \textbf{Blocking dependency:} In a task-chain system, where tasks depends on the result of others to compute valid outputs, we assume data dependency to be a blocking constraint for a task execution: a task $\tau_j$ that depends on the result of a task $\tau_i$ will have its release delayed waiting for the result of $\tau_i$ to be available. Mathematically, if the job $m$ of $\tau_j$ depends on the result of job $k$ of $\tau_i$, $s_{j,m} \geq f_{i,k}+L_{ij}$.
    \item \textbf{Independent Release:} Tasks with no dependency, i.e., $pred(\tau_i) = \emptyset$, are not bound by precedence constraints, and may be scheduled at any offset $\Phi_i | f_{i,k} < s_{j,k} ~\forall \tau_j \in succ(\tau_i)$. If $succ(\tau_i) = \emptyset$, then $\Phi_i$ is only limited by $\tau_i$ being able to complete before $T_i$ in the worst case scenario, i.e., $\Phi_i \leq T_i-C_i$.
    \item \textbf{Localized Data Freshness (Non-Transitive Data Age):} 
    We assume the data freshness attribute is reset at each processing stage. The age of a derived sample produced by $\tau_j$ is relative to $\tau_j$'s finishing time, not its inputs finishing time. \textit{Rationale:} the data freshness of a task \textit{output} represents the useful life of the current computation result.
    \item \textbf{Actuation-Driven Phasing:} 
    The system timeline is anchored to the actuation requirements. We assume the first actuation dictates the periodicity start ($t=0$).
    \textit{Rationale:} The initial period is treated as a safe "booting phase" where the system is permitted to pre-load the pipeline. This allows high-latency data (where $E_{i,j} > T_i$) to be sampled in the negative timeline (pre-boot) or during the first cycle, ensuring that steady-state data freshness constraints are met from the very first effective actuation.
\end{itemize}

We operate under the following assumptions to ensure system predictability:
\begin{itemize}
    \item \textbf{Well-Behaved Pipeline:} $\Phi_i + C_i + L_{ij} < T_i$. This ensures the entire lifecycle of a job, from its offset-adjusted release to the arrival of its results at a consumer, is contained within one period $T_i$. In CPS, this prevents "backlog accumulation," where a producer task would otherwise have pending work from previous periods interfering with current deadlines.
    \item \textbf{Phase Boundedness:} By consequence, $\Phi_i < T_i$ and $\Phi_j < T_j$. This constrains the relative phase shifts to the task periods, preventing arbitrary delays that would complicate schedulability.
\end{itemize}

We adopt the following assumptions to abstract the hardware/software interface:
\begin{itemize}
    \item 
    \item \textbf{Deterministic Data Transmission Latency:} 
    We assume a directly connected topology where the transmission time of a result from a task $\tau_i$ to $\tau_j$ is constant and deterministic, given by a latency $L_{ij}$. This abstraction is envisioned to encompass both shared memory and time-sensitive networks (TSN) or switched Ethernet with static forwarding tables in modern CPS\cite{autosar2022timing}.
    \item \textbf{Atomic Execution-Communication interleave:} 
    We model the interleave between processing and communication phases as atomic operations. When a job of task $\tau_i$ completes its computation, it immediately flushes data to the other tasks (which might take up to $L_{ij}$ units of time). This simplifies the interference model by treating the computation finishing time $f_{i,k}$ and communication latency $L_{ij}$ as a contiguous block in the worst-case analysis.
\end{itemize}

\begin{remark} \textbf{Minimum Data Freshness:} $E_{i,j} \geq C_i + L_{ji}$, otherwise, at least for the worst cases of computation and communication latency the data will never be fresh.
\end{remark}

\subsection{Data Age}
Consider $\tau_i \in \mathcal{T}$ and $\tau_j \in \mathcal{T}$, and $\tau_j \in succ(\tau_i)$. To determine the index $n$ of the freshest sample of task $\tau_i$ available for consumption by the $k$-th job of task $\tau_j$, we complement the previous definition with the \textbf{Arrival at Consumer:} $A_{i,n} = f_{i,n} + L_{ij} = (n-1)T_i + \Phi_i + C_i + L_{ij}$. The freshest data available for the consumer is the job $\tau_{i,n}$ that satisfies the condition $A_{i,n} \leq s_{j,k}$. Substituting the definition above:
\begin{equation}
    (n-1)T_i + \Phi_i + C_i + L_{ij} \leq (k-1)T_j + \Phi_j
\end{equation}
Dividing by $T_i$ and solving for $n$:
\begin{equation}
    n \leq \frac{(k-1)T_j + (\Phi_j - \Phi_i - C_i - L_{ij})}{T_i} + 1
\end{equation}
Given that $n$ must be an integer, the index is given by:
\begin{equation}
    n(k) = \left\lfloor \frac{(k-1)T_j + \Phi_j - \Phi_i - C_i - L_{ij}}{T_i} \right\rfloor + 1
\end{equation}
Defining the \textbf{Period Ratio} $\alpha = \frac{T_j}{T_i}$ and the \textbf{Phase Constant} $\beta = \frac{\Phi_j - (\Phi_i + C_i + L_{ij})}{T_i}$, we obtain:
\begin{equation}
    n(k) = \left\lfloor (k-1)\alpha + \beta \right\rfloor + 1
\end{equation}

\textbf{The Period Ratio ($\alpha$)} governs the \textit{rate} of index progression. It dictates how many jobs of $i$ occur per job of $j$. If $\alpha > 1$, the producer is faster; if $\alpha < 1$, the producer is slower. \textbf{The Phase Constant ($\beta$)} governs the \textit{initial alignment}. Given the constraint $\Phi_i + C_i + L_{ij} < T_i$, $\beta$ is bounded such that $\lfloor \beta \rfloor \in \{-1, 0\}$ for synchronous systems ($T_i=T_j$). It essentially shifts the baseline index from which $n$ evolves. Thus, the expected values for $n$ based on the Period ratio are:
\begin{itemize}
    \item \textbf{Synchronous ($\alpha = 1$):} $n$ increments linearly with $k$, with $\beta$ determining the index lag (either $n=k$ or $n=k-1$, depending on $\Phi$). By grounding on the assumptions for \textit{Actuation-Driven Phasing}, $n-=k$ always.
    \item \textbf{Undersampling (Producer is slow, $\alpha < 1$):} The consumer $\tau_j$ executes faster than the producer $\tau_i$. The index $n$ increases slower than $k$. The consumer must reuse the same sample of $\tau_i$ across multiple jobs, leading to data aging, which should be encompassed by a sufficiently large data freshness constraint in the dependency relationship, otherwise, data freshness will never be respected.
    \item \textbf{Oversampling (Producer is fast, $\alpha > 1$):} The consumer $\tau_j$ executes slower than the producer $\tau_i$. The index $n$ increases faster than $k$. The consumer skips intermediate samples of $\tau_i$, selecting the most recent job $n$ to maintain freshness (defined by both $\alpha$ and $\beta$).
\end{itemize}

Considering the definition of $n(k)$ for a given data dependency $(\tau_i,\tau_j) \in \mathcal{E}$, we adopt the following definition of data age for multi-rate chains:
\begin{definition}\textbf{Data Age:} For a job $\tau_{j,k}$ that finishes at $f_{j,k}$, and a job $\tau_{i,n(k)}$, that $\tau_j$ depends on its result, that finishes its execution at $f_{i,n(k)}$, the age of the sample $n(k)$ of $\tau_i$ consumed by $\tau_{j,k}$ is the time elapsed between $\tau_{i,n(k)}$ and $\tau_{j,k}$ completion:
\begin{equation}
    Age_j(\tau_{i,n(k)}) = f_{j,k} - f_{i,n(k)}
\end{equation}
\end{definition}
Therefore, network delays ($L_{ij}$) and consumer scheduling and computation time ($f_{j,m}$) consume the data freshness budget $E_{i,j}$. This definition holds as long as the assumption for \textit{Blocking dependency} holds. We present the following definition for a data freshness compliant real-time application:
\begin{definition}\textbf{Data Freshness Compliance:} Given a real-time application modeled as a DAG $G = (\mathcal{T}, \mathcal{E})$, the system respects data freshness constraints when:
\begin{equation}
    \forall \tau_j \in \mathcal{T}, \forall \tau_i \in \Omega_j ~\land k \in J_j : Age_j(\tau_{i,n(k)}) \leq E_{i,j}
\end{equation}
\end{definition}
Note that, since we only check for the respective $\tau_{i,n(k)}$, we do not care for the average data age of unused samples.

The objective of the proposed Start-Time Optimization is to compute the vector $\mathbf{\Phi} = \{ \Phi_1, \dots, \Phi_n \}$ such that tasks with tight $E_{i,j}$ are delayed (increasing $\Phi_i$) to execute Just-In-Time for their consumers while attending to their deadline constraints as well, thereby minimizing the data age.

% \vspace{1mm}
% \begin{remark}
%     Considering a data dependency $e_{i,j}$, in the worst case $s_{j,k} = D_i-(C_i+C_j+L_{ji})$ and $f_{j,k} = s_{j,k} + C_i$, and $f_{i,k} = D_i$.
% \end{remark}

\section{Adapting Start Offset to Optimize Data Freshness in Task Chains}\label{sec:proposal}

Traditional real-time scheduling focuses on task urgency, where the most constrained tasks (shortest deadlines) are prioritized to execute as early as possible. However, in sensor-fusion chains, earliest deadline first execution is often the primary driver of data staleness. When a high-frequency sensor with strict data freshness completes early, it must wait for high-latency support tasks (e.g., vision processing) to complete before fusion can occur. During this synchronization gap, the fresh data ages, often exceeding its data freshness bound before it is ever consumed.

To address this, we propose a shift toward \textit{Just-In-Time} scheduling. By calculating static release offsets ($\Phi$) based on data freshness constraint rather than task urgency, we can anchor the execution of the entire chain to the slowest bottleneck and delay the sampling of fast sensors. This ensures that all inputs arrive at the fusion node at the peak of their data freshness.

\subsection{Motivating Example: The AEB Sensor Fusion Chain}
To illustrate the impact of this schedule alignment, we analyze a simplified Automated Emergency Braking (AEB) system under both single-core and multi-core scenarios.

\textbf{System Parameters}
\begin{itemize}
    \item \textbf{All tasks follow the same global period ($T_{ctrl}$):} $20ms$, and for simplicity of the example, we assume $L_{imu,ctrl}=L_{vis,ctrl}=1ms$
    \item \textbf{Consumer Task ($\tau_{ctrl}$):} Brake Controller ($C_{ctrl} = 1ms, E_{ctrl} = 20ms$).
    \item \textbf{Input 1 - IMU ($\tau_{imu}$):} Strict Constraint ($C_{imu} = 2ms$, $E_{ctrl~imu} = 5ms$).
    \item \textbf{Input 2 - Vision ($\tau_{vis}$):} Loose Constraint ($C_{vis} = 10ms$, $E_{ctrl~vis} = 20ms$).
\end{itemize}

\textbf{Scenario 1: Single-Core Serialization} \\
In a single-core system, a naive Least Laxity First Scheduler, where priority equals to the freshness constraint, will incur on the following order $\tau_{imu} \to \tau_{vis} \to \tau_{ctrl}$), resulting in $\tau_{imu}$ finishing at $t=2ms$. The data then ages $10ms$ for $\tau_{vis}$ to finish. At the conclusion of $\tau_{ctrl}$, $f_{ctrl,1} = 14ms$), the IMU data age is:
\begin{gather}
    Age_{imu} = f_{ctrl,1} - f_{imu,1} = 14ms - 2ms = 12ms > E_{imu,ctrl}\notag\\\notag \implies \text{\textbf{FAILURE}}
\end{gather}

By inverting this to a JIT-based scheduling ($Order: \tau_{vis} \to \tau_{imu} \to \tau_{ctrl}$), the loose task fills the initial gap. $\tau_{imu}$ finishes at $f_{imu,1}=12ms$ and is consumed at $t=14ms$.
\begin{gather}
    Age_{imu} = 14ms - 12ms = 2ms \leq E_{imu,ctrl} \notag\\ \implies \text{\textbf{SUCCESS}}\notag
\end{gather}

\textbf{Scenario 2: Multi-Core Global Scheduling} \\
On a multi-core platform with at least two cores, a naive Least Laxity First Scheduler releases both sensors at $t=0$. $\tau_{imu}$ finishes on Core 1 at $t=2ms$, while $\tau_{vis}$ finishes on Core 2 at $t=10ms$. Synchronization for $\tau_{ctrl}$ occurs at $t=12ms$:
\begin{gather}
    Age_{imu} = 12ms - 2ms = 10ms > E_{imu,ctrl} \implies \text{\textbf{FAILURE}}\notag
\end{gather}

Our proposed \textbf{Offset-based} approach anchors the chain to the vision bottleneck ($t=11ms+1ms+1ms$, $f_{vis,k}+L_{vis,k}+C_{ctrl}$) and calculates a release offset for the IMU: $\Phi_{imu} = 13ms - 5ms = 8ms$ ($f_{vis,1}-(E_{imu,ctrl})$. By delaying the IMU release until $t=8ms$, it completes at $t=10ms$:
\begin{equation}
    Age_{imu} = 13ms - 10ms = 3ms \leq E_{imu,ctrl} \implies \text{\textbf{SUCCESS}}\notag
\end{equation}

The remainder of this section formalizes the methodology required to implement this data freshness-aware scheduling across Directed Acyclic Graphs (DAGs):
\begin{itemize}
    \item \textbf{Dependency in Task-chains:} Formalizing the relationship between producers and consumers.
    \item \textbf{Demand-Driven Period Derivation:} Determining the optimal $T$ for producers based on consumer frequency.
    \item \textbf{Dominant Path Decomposition:} Identifying the bottleneck paths that dictate the anchor points of the schedule.
    \item \textbf{Task Prioritization (Linear and Non-Dominant):} Proposal for assigning offsets in both simple chains and complex branching structures.
    \item \textbf{Shared Producer (Multi-Rate JIT):} Resolving offset conflicts when a single producer serves multiple consumers with varying rates and data freshness windows.
\end{itemize}

\subsection{Dependency in Task-chains}

In contrast to traditional sensing-driven models, we employ a \textit{Demand-Driven} approach. The temporal attributes of a producer task $\tau_i$ are not arbitrary; they are derived constraints dictated by the requirements of its consumer set $succ(\tau_i)$. We define the derivation rules for the producer's period $T_i$ and required data freshness $E_{i,j}$ as follows:

\textbf{1. Linear Dependency (1:1 Case):}
Consider a simple chain where task $\tau_i$ produces data solely for a single consumer $\tau_j$. The producer would be optimized in terms of processing time used if it could adapt its rate to match the consumption rhythm of the consumer.
\begin{itemize}
    \item \textbf{Derived Period:} To guarantee fresh data availability for every execution of the consumer without wasteful oversampling, the producer can inherit the consumer's demands, either producing a new sample at the consumer rate (i.e., $T_j \geq E_{i,j}$) or according to the data freshness constraint (i.e., when $T_j < E_{i,j}$).
    \begin{equation}
        T_i = T_j%
        %T_i = \max(T_j, E_{i,j})
    \end{equation}
\end{itemize}

\textbf{2. Shared Producer / Divergent Flow (1:N Case):}
Consider a Producer $\tau_i$ that publishes data to a set of multiple consumers $succ(\tau_i) = \{ \tau_1, \dots, \tau_N \}$. The producer must adapt its period to properly attend \textit{all} of its consumers simultaneously.
\begin{itemize}
    \item \textbf{Derived Period (Harmonicity):} To maintain synchronization with multiple consumers, the producer must operate at a rate that is harmonically compatible with the entire consumer set:
    \begin{equation}
        T_i = \text{GCD}\left( \{ T_j \mid \forall \tau_j \in succ(\tau_i) \} \right)
        %T_i = \text{GCD}\left( \{ \max(T_c,E_{i,j}) \mid \forall \tau_j \in succ(\tau_i) \} \right)
    \end{equation}
\end{itemize}

\begin{remark} As $\text{GCD}$ is used to guarantee sampling at a rate that is a common divisor of all periods, the choice of periods and data freshness constraint that are not in phase for multiple actuation sharing the same dependency might thrash system performance.
\end{remark}

\begin{remark}
This backward propagation of periodicity ensures that the Actuator (the end of a data path) drives the timing definition. The Actuator's period $T_{head}$ is determined by the physical plant dynamics, and his rhythm propagates upstream to the sensors in accordance with their data freshness constraint.
\end{remark}

\subsection{Dominant Path Decomposition}
We define the dominance relationships based on the \textit{data freshness} constraints for dominant input selection, and \textit{Period} ($T$) for output selection.

\begin{definition}[Critical Predecessor ($previous(\tau_i)$)]
    Identifies the dominant dependency for task $\tau_i$. The dominant predecessor is the task $\tau_j$ whose data has the strictest data freshness constraint (smallest data freshness window) imposed by $\tau$. This defines the critical data freshness path.
    \begin{equation}
        previous(\tau_i) = 
        \begin{cases} 
            \underset{\tau_j \in pred(\tau_i)}{\arg\min} \{ E_{i,j} \} & \text{if } pred(\tau) \neq \emptyset \\
            \emptyset & \text{if } pred(\tau_i) = \emptyset
        \end{cases}
    \end{equation}
\end{definition}

We define the \textbf{Dominant Task Chain} as a path through the Task Dependency Tree (TDT) rooted at a specific Actuator Task, denoted in this TDT as $\tau_{0}$, where every node $\tau_i$ in the path is the critical predecessor of the previous node $\tau_{i-1}$ in the path (i.e., $\tau_i = previous(\tau_{i-1})$).

\subsection{Task Prioritization - Case 1: Linear Task Chain}
For a linear task chain $\mathcal{P}$ executing on a single processor, the explicit modeling of data dependencies can be simplified. We reduce the precedence-constrained scheduling problem to a standard \textbf{Earliest Deadline First (EDF)} problem by deriving \textit{Effective Deadlines} for each task, modeled as the priority in our task model, e.g., $P_i = D^{eff}_i$.

In a linear chain, a predecessor $\tau_i$ must finish early enough to allow all subsequent tasks $\{\tau_{i+1}, \dots, \tau_{N}\}$ to execute and communicate before the global chain deadline. Therefore, the absolute deadline of a task is not its period $T_i$, but a stricter constraint derived from the demand of its successors.

Let the task chain consist of $N$ tasks, where $\tau_1$ is the chain head (e.g., an actuator task with no consumers) and $\tau_N$ is the last task in the chain (e.g., a sensor task with no dependencies).

\textbf{1. Base Case (Chain Head):}
The actuator must complete by the end of the period.
\begin{equation}
    D^{eff}_1 = T_1
\end{equation}

\textbf{2. Recursive Step (Upstream Propagation):}
For any task $\tau_i$ ($1 < i < N$), the deadline is the effective deadline of its successor $\tau_{i+1}$ minus the execution time of the successor and the communication latency between them.
\begin{equation}
    D^{eff}_i = D^{eff}_{i-1} - C_{i-1} - L_{i, i-1}
\end{equation}

\textbf{3. General Closed Form:}
Substituting recursively, the effective deadline for any task $\tau_i$ is the global period reduced by the cumulative worst-case execution and transmission time of the downstream path:
\begin{equation}
    D^{eff}_i = T - \sum_{j=i-1}^{1} \left( C_j + L_{j-1, j} \right)
\end{equation}

By assigning these derived deadlines, we guarantee that:
\begin{equation}
    D^{eff}_1 < D^{eff}_2 < \dots < D^{eff}_N
\end{equation}
When scheduled under EDF:
\begin{itemize}
    \item At the start of the period ($t=0$), all tasks are released.
    \item The scheduler selects the task with the smallest deadline, which is always the task with no dependencies $\tau_N$.
    \item Once $\tau_N$ completes, $\tau_{N-1}$ becomes the runnable task with the earliest deadline.
    \item This naturally enforces the topological order $\tau_N \to \tau_1$ purely through deadline priorities, satisfying the dependency constraints without explicit synchronization primitives.
\end{itemize}

Finally, data freshness is encompassed by choosing the minimum between the effective deadline calculated and the data freshness constraint issued from $\tau_{i-1}$ to its predecessor task $\tau_i$:
\begin{definition}[Data Freshness Adjustment]
If the data freshness constraint $E_{i,i-1}$ is tighter than the structural slack, the deadline must be further clamped. The generalized effective deadline is:
\begin{equation}
    D^{eff}_i = \min \left( D^{eff}_{i-1} - C_{i-1} - L_{i, i-1}, \quad r_i + E_{i,i-1} \right)
\end{equation}
\end{definition}

\subsubsection{Schedulability in Global EDF}
Standard Global EDF proofs rely on the property that the scheduler is \textit{work-conserving} (a core is never idle if a task is ready). A suspended task is "open" but not "ready," creating an anomaly where the processor idles despite pending workload.
\textbf{Consequence:} It is known that self-suspensions induce a "jitter" effect that is harder to schedule than pure computation. The standard G-EDF schedulability test fails because:
\begin{equation}
    \text{Demand}(t) + \text{SuspensionGap}(t) > \text{Supply}(t)
\end{equation}
The parameter $t \in \mathbb{R}^+$ denotes a continuous time interval. For any task $\tau_i$, the demand $Demand(\tau_i, t)$ is the sum of the execution times of all jobs of $\tau_i$ whose release times and absolute deadlines lie within $[0, t]$, while $Supply(t) = m \cdot t$, where $m$ is the number of cores.

The total system demand over interval $t$ is then:
\begin{equation}
    \sum_{\tau_i \in \mathcal{T}} \left( \left\lfloor \frac{t - \Phi_i}{T_i} \right\rfloor + 1 \right) \cdot C_i + \text{SuspensionGap}(t)
\end{equation}
Where the \textbf{SuspensionGap} is the sum of all idle intervals on $m$ cores where no critical task can be scheduled due to data freshness-imposed offsets.

To restore the validity of the EDF proof, we must move from "Independent Task Analysis" to \textbf{Transaction-Based Analysis} employing the offsets $\Phi_i$. %Note that this is only a logical artifact to evaluate schedulability of the system considering worst-case analysis.%, and the real system will not suffer from these offsets and will idle the CPU instead.

Since we are considering a simple Linear Task Chain, we enforce that each task is only released when the previous one completes (considering its WCET and communication latency), i.e., $\Phi_{i-1} = C_{i}+L_{i,i-1}$. By enforcing this release offset, we ensure that:
\begin{enumerate}
    \item \textbf{Independence Restoration:} When $\tau_{i-1}$ is finally released at $t = \Phi_{i-1}$, the data is already available. $\tau_{i-1}$ never waits.%, i.e., $C'_{cons} = C_{cons}$.
    \item \textbf{Validity of EDF:} The demand of $\tau_{i-1}$ is shifted in time. It contributes to the system load only over the interval $[\Phi_{i-1}, D_{i-1}]$.
\end{enumerate}

The schedulability check must be updated to an \textbf{Offset-Aware Test}:
\begin{equation}
    \forall t > 0, \quad \sum_{\tau_i \in \mathcal{T_i}} \max\left(0, \left\lfloor \frac{t - \Phi_i}{T} \right\rfloor + 1 \right) \cdot C_i \leq m \cdot t
\end{equation}
where $m$ is the number of processors. The test must hold for all $t > 0$ to ensure the system is schedulable under all conditions. In practice, $t$ is checked at task arrivals and deadlines up to the Least Common Multiple (LCM) of all task periods. The term $\lfloor \frac{t - \Phi_i}{T_i} \rfloor + 1$ counts how many instances (jobs) of task $\tau_i$ have been released within the window $t$, accounting for the delay introduced by the offset $\Phi_i$. This formula correctly accounts for the fact that downstream tasks cannot demand CPU time at the beginning of the period.

\subsection{Task Prioritization - Case 2: Handling Non-Dominant Branches (single consumer, multiple producers)}

Since we are not dealing with a linear task-chain anymore, we will adopt the notation $\tau_{cons}$ to represent the consumer task, and $\tau_{prod_k}$ to represent the $k$-th task $\tau_{cons}$ depends on.
When a consumer task $\tau_{cons}$ requires synchronized data from multiple predecessor branches (e.g., multiple inputs), the schedule must ensure that \textit{all} input data remains fresh up to the conclusion of the $k$-th job of the consumer, i.e., $f_{cons,k}$.

This problem is universally solved by identifying the system's temporal \textbf{Anchor} (latest arriving input) and aligning all other branches relative to it. This approach applies equally to single-core (ordering) and multi-core (offset) scenarios.

By the assumptions assumed in this paper, a consumer task $\tau_{cons}$ cannot start until all of its inputs are available (and its results are only valid if the inputs were fresh). The release of the $k$-th job of a task $\tau_{cons}$, $r_{cons,k}$, is determined by the "Bottleneck Branch", the predecessor path with the longest latency (sum of $\Phi + C + L$ for all tasks in the path). The anchor time relative to the period start of the $k$-th job of a task $\tau_{cons}$ is then the sum of the start time with the worst-case execution time:
\begin{equation}
    f_{cons,k} = \max_{\tau_{prod_i} \in pred(\tau_{cons})} \left(\Phi_i + C_i + L_{i, cons} \right) + C_{cons}
\end{equation}
By applying this rule iteratively, from the nodes that depends on tasks with no dependencies, upstream, we obtain calculate the worst-case start time for all nodes in the dependency tree.

Once the anchor $f_{cons,k}$ is established for a task $\tau_{cons}$, we calculate the required timing for every other predecessor $\tau_{prod_i}$ to ensure its data does not expire before $f_{cons,k}$, given as the \textbf{Latest Safe Start Time for this consumer ($LST_{prod_i,cons}$)}:
\begin{equation}
    LST_{prod_i,cons} = f_{cons,k} - E_{prod_i,cons}
\end{equation}

The scheduler's goal is to align the start time $s_{prod_i,k}$ with $LST_{prod_i,cons}$ as close as possible. The gap between the earliest possible start ($0$) and the required start ($LST_{prod_i,cons}$) represents the \textbf{Temporal Slack}.

How we utilize this slack depends on the hardware resource constraints, but the mathematical target ($LST_{prod_i,cons}$) remains constant.
\subsubsection{Case 2.1: Global Scheduling (Multi-Core)}
\begin{itemize}
    \item \textbf{Mechanism:} \textit{Explicit Offset Injection.}
    \item Since cores are available in parallel, we simply delay the release of the fast/strict task $\tau_{prod_i}$ until the last safe moment to minimize data age and buffer residence.
    \begin{equation}
        \Phi_{prod_i} = LST_{prod_i,cons}
    \end{equation}
    \item \textbf{Result:} The task sleeps during the slack time, saving system bandwidth or allowing lower-priority background tasks to run. The effective deadline of these task is then set to $D^{eff}_{prod_i} = f_{cons,k}$, the anchor.
\end{itemize}

\subsubsection{Case 2.2: Serial Scheduling (Single-Core)}
\begin{itemize}
    \item \textbf{Mechanism:} \textit{Effective Deadline Ordering.}
    \item To ensure the most fragile data (tight data freshness) is generated last (immediately before consumption) we treat the data freshness constraints as a \textbf{Virtual Dependency Chain}.
    \item We conceptually order the priority by data freshness, the loose is the constraint (higher data freshness value) the highest is the priority (executes first), i.e., ($\tau_{N} \to \ldots \to \tau_{2} \to \tau_{1}$), where $\tau_1 = \tau_{cons}$, and $E_{\tau_{i+1},\tau_1} \geq  E_{\tau_{i},\tau_1}$.
    \item \textbf{Deadline Derivation:} Using our recursive effective deadline definition:
    \begin{gather}
        P_i = D^{eff}_{i} = D^{eff}_{i-1} - C_{i}\text{, with}\\
        P_1 = D^{eff}_{1} = T_1
    \end{gather}
    \item \textbf{Result:} Since $D^{eff}_{N} < D^{eff}_{N-1}$ then $E_{\tau_{N},\tau_1} \geq  E_{\tau_{N-1},\tau_1}$, an EDF scheduler naturally selects the Loose task first. It executes earlier in the frame, allowing the Strict task to execute later (closer to the deadline), thereby minimizing the age of the strict data at the moment of consumption.
\end{itemize}

\begin{remark}
In both cases, the \textbf{Offset-Aware Schedulability Test} validates the system.
\begin{itemize}
    \item For Multi-Core, $\Phi_{prod_i}$ reduces the interference window, and when combined with the effective deadline set to the start of the next task, allows for ordering even in G-EDF.
    \item For Single-Core, the ordering ensures that no task effectively "blocks" another beyond the allowable data freshness window.
\end{itemize}
\end{remark}

\subsection{Shared Producer: Multi-Rate JIT Scheduling}

We consider a shared producer $\tau_{shared}$ feeding multiple consumer chains $succ(\tau_{shared}) = \{ \tau_1, \dots, \tau_N \}$. A conflict arises when a static offset $\Phi_{shared}$ chosen to satisfy a strict consumer (requires freshest samples, leading to late-frame production) causes the task to finish after the consumption window of a less strict consumer, forcing it to consume stale data. Formally, if for a consumer $\tau_i \in succ(\tau_{shared})$, $\Phi_i < \Phi_{shared} + C_{shared}$, the consumer $\tau_i$ attempts to execute before $\tau_{shared}$ has produced new data. It is forced to read the data from a previous job, which is likely expired ($Age \approx T_{shared} + \Phi_{shared} + C_{shared}$), leading to a system-wide deadline or freshness failure.

To define the validity windows, we establish the following temporal anchors:
\begin{itemize}
    \item $f_{i,1}$: The finish time of the first job ($k=1$) of consumer $\tau_i$. The algorithm solves the problem for the first job, generalizing the solution to other jobs as the system is periodic, simplifying the interval verifications.
    \item $LB_i = f_{i,1} - E_{shared,i}$: the earliest time (lower bound) a data sample can be produced while remaining fresh until the consumer's execution.
    \item $UB_i = f_{i,1} - C_{shared} - L_{shared,i} - C_i$: the latest time (upper bound) a data sample can be produced to ensure it is finished and transmitted in time for the consumer's execution.
\end{itemize}

Our objective is to identify a static global offset $\Phi_{shared}$ that satisfies all consumer constraints via \textbf{Static Offset Propagation}, an algorithm that searches for a single offset that satisfies all consumers throughout multiple jobs of $\tau_{shared}$ within the consuming interval of consumers.

\begin{remark} \textbf{The Cascade Effect and Complexity} Adjusting $\Phi_{shared}$ to satisfy one consumer may inadvertently invalidate the freshness constraints of another. Trying to recursively adjust anchors in the critical path of other consumers might cause a \textbf{Cascade Effect}. Because each adjustment to $\Phi_{shared}$ can recursively necessitate a re-synchronization of all predecessors in the chain, the search space for a global static offset grows exponentially. Finding an optimal $\Phi_{shared}$ that avoids this recursive cascade, or proving no such offset exists, is a combinatorial optimization problem.
\end{remark}

When a single $\Phi_{shared}$ is insufficient to satisfy the intersection of all freshness windows $\bigcap_{i} [LB_i, UB_i]$, we apply \textbf{Hyperperiod Decomposition}. We decompose $\tau_{shared}$ into $N = HP_{shared}/T_{shared}$\footnote{$HP_{shared} = LCM_(T_i ~\forall \tau_i \in succ(\tau_{shared}))$} sub-tasks, each with a unique fixed offset $\Phi_{shared,k}$ tailored to the consumer active in that hyperperiod slot. This preserves EDF schedulability by maintaining strict periodicity without job-specific offsets, while avoiding the combinatorial complexity.

Algorithm~\ref{alg:consensus} employs back-propagation to determine a valid $\Phi_{shared}$. The consensus algorithm navigates the search space of $\Phi_{shared}$ by iteratively validating the availability of data samples for every consumer $\tau_i$.

\begin{enumerate}
    \item \textbf{Initialization:} We start with an optimistic offset based on the most constrained path (the highest frequency consumer).
    \item \textbf{Validation:} For every consumer $\tau_i$, the algorithm generates a set of potential release times $\mathcal{I}$ for all jobs of $\tau_{shared}$ that could logically be consumed within the consumer's period.
    \item \textbf{Back-Propagation and Adjustment:} If no sample $t \in \mathcal{I}$ falls within the required freshness window $[LB_i, UB_i]$, a conflict is declared. The algorithm employs a \textit{back-propagation} step: it adjusts $\Phi_{shared}$ to "push" the producer's execution closer to the consumer's start time.
    \item \textbf{Rollback Logic:} If the adjustment forces the offset beyond $T_{shared}$, it implies no single static offset can satisfy the global constraints. At this point, the algorithm rolls back and triggers a failure signal, which initiates the \textbf{Hyperperiod Decomposition} to ensure schedulability.
\end{enumerate}

\begin{algorithm}
\caption{Shared Producer Consensus Search}\label{alg:consensus}
\begin{algorithmic}[1]
\Require $\tau_{shared}$, $succ(\tau_{shared})$
\Ensure Static Offset $\Phi_{shared}$ or Failure
\State $ord \gets \text{sort } succ(\tau_{shared}) \text{ by ascending } T_i$
\State $\Phi_{shared} \gets f_{ord[0], 1} - E_{shared,ord[0]}$
\While{not Success}
    \State $ConflictFound \gets \text{False}$
    \For{$\tau_i \in ord$}
        \State $LB_i \gets f_{i,1} - E_{shared,i}$ 
        \State $UB_i \gets f_{i,1} - C_{shared} - L_{shared,i} - C_i$ 
        \State $\mathcal{I} \gets \{j \cdot T_{shared} + \Phi_{shared} \mid j \in [0, \frac{T_i}{T_{shared}}-1]\}$
        \If{$\nexists t \in \mathcal{I} \mid t \in [LB_i, UB_i]$}
            \State $\Phi_{shared} \gets \Phi_{shared} + 1$
            \State $ConflictFound \gets \text{True}, \textbf{break}$
        \EndIf
    \EndFor
    \If{\textbf{not} $ConflictFound$}
        \State $Success \gets \text{True}$
    \EndIf
    \If{$\Phi_{shared} > T_{shared}$}
        \State \Return \textit{Failure: Invoke Hyperperiod Decomposition}
    \EndIf
\EndWhile
\State \Return $\Phi_{shared}$
\end{algorithmic}
\end{algorithm}

The increment $\Phi_{shared} \gets \Phi_{shared} + 1$ ensures a monotonic search toward the upper bound $UB_i$. Given that the freshness constraint $Age \leq E_{ij}$ is monotonically decreasing with respect to $\Phi_{shared}$ (within the interval $[LB_i, UB_i]$), the algorithm converges to the supremum of the feasible completion time, effectively minimizing the consumed data age.

\subsection{Consistency Preservation and Hyper-period analysis}
We ensure \textbf{Phase-Shift Invariance}: if $\Phi_{shared}$ acts as an anchor for $\tau_i$, we shift all other producers of $\tau_i$ by the same $\Delta$ to preserve relative data age. If $\Phi_{shared}$ is not the anchor, offsets remain unchanged, preserving path-specific synchronization. Since $\mathcal{I}$ considers all jobs of $\tau_{shared}$ that could be consumed by $\tau_{i}$ ($\mathcal{I} =[0,\frac{T_i}{T_{shared}}]$), if we end up with $\Phi_{shared} \geq T_{shared}$, there is no offset that would satisfy all consumers. Therefore, an hyperperiod decomposition is required to enforce different offsets to different jobs of $\tau_{shared}$.

\subsubsection{Hyperperiod Decomposition}
In cases where no single $\Phi_{shared}$ satisfies the intersection of all consumer's data freshness windows ($\bigcap_{i} [LB_i, UB_i]$), the task $\tau_{shared}$ is decomposed into $N = HP_{shared}/T_{shared}$ distinct periodic tasks. Each task $\tau_{shared,k} \mid 1 < k < N$ is assigned a unique fixed offset $\Phi_{shared,k}$ tailored to the specific consumer active in that hyperperiod slot. This preserves the periodicity required for the EDF Schedulability Test while satisfying heterogeneous data freshness requirements.

\section{Proof of Schedulability Capacity Preservation and Data Age Minimization}\label{sec:proof}
This section presents the proof that our offset based solution preserves EDF Schedulability Capacity, which is done through demand bound function \cite{Bini:2010} and properly handling the \textit{work-conserving property} when tackling self-suspension of tasks, and the data age minimization.

\subsection{Theorem: Asynchronous Dominance}
\begin{theorem}
Let $\mathcal{T} = \{\tau_1, \dots, \tau_n\}$ be a set of periodic tasks with implicit deadlines $D_i = T_i$. The introduction of static offsets $\Phi_i \geq 0$ to satisfy data freshness constraints does not reduce the maximum schedulable utilization bound of the system under Global EDF.
\end{theorem}

\subsubsection{Proof by Demand Bound Function (DBF)}
The Demand Bound Function $DBF(\tau_i, t)$ defines the maximum execution requirement of task $\tau_i$ in any interval of length $t$. For a synchronous task ($\Phi_i = 0$):
\begin{equation}
    DBF_{sync}(\tau_i, t) = \max \left( 0, \left\lfloor \frac{t - D_i}{T_i} \right\rfloor + 1 \right) \cdot C_i
\end{equation}

When an offset $\Phi_i$ is introduced, the first job is delayed, shifting the demand curve to the right:
\begin{equation}
    DBF_{async}(\tau_i, t) = \max \left( 0, \left\lfloor \frac{t - (\Phi_i + D_i)}{T_i} \right\rfloor + 1 \right) \cdot C_i
\end{equation}

\textbf{Step 1: Interval Dominance.}
For any $t > 0$, since $\Phi_i \geq 0$, it follows that $t - (\Phi_i + D_i) \leq t - D_i$. Consequently, $DBF_{async}(\tau_i, t) \leq DBF_{sync}(\tau_i, t)$. Because the individual task demand is non-increasing, the total system demand $\sum DBF_{async}(t)$ does not exceed the synchronous demand bound. While asynchronous arrival patterns may change the distribution of peak demand, they do not increase the total workload density over the hyperperiod.

\textbf{Step 2: Utilization Invariance.}
The system utilization $U$ is defined by the long-term work density:
\begin{equation}
    U = \sum_{i=1}^{n} \lim_{t \to \infty} \frac{DBF(\tau_i, t)}{t} = \sum_{i=1}^{n} \frac{C_i}{T_i}
\end{equation}
As $\Phi_i$ is constant, $\lim_{t \to \infty} \frac{DBF_{async}(t)}{t} = \lim_{t \to \infty} \frac{DBF_{sync}(t)}{t}$. Thus, $U_{async} = U_{sync}$.

\subsubsection{Handling the "Self-Suspension" Artifact}
Unlike dynamic self-suspension, where tasks block themselves based on internal state changes, our $\Phi_i$ represents a \textit{fixed release delay}.
\begin{itemize}
    \item \textbf{Work-Conserving Property:} Under G-EDF, the scheduler is work-conserving. If core $m_j$ is idle during $[0, \Phi_i]$, the resource remains available to other ready tasks.
    \item \textbf{Density Invariance:} We preserve the relative deadline $D_i = T_i$. Because both release times and deadlines are shifted by $\Phi_i$, the task density $\delta_i = C_i/D_i$ remains invariant.
\end{itemize} 

\subsubsection{Proof Conclusion}
Since the total demand bound is non-increasing and utilization remains invariant, any task set schedulable at $U \leq m$ without offsets remains schedulable with freshness-aware offsets.\hfill $\square$

\subsection{Proof of Data Age Minimization}
The algorithm exhibits a greedy convergence property. By initializing at the lower bound $LB_i$ and incrementing $\Phi_{shared}$ until the feasibility condition $f_{shared} \in [LB_i, UB_i]$ is met, the search path is strictly monotonic. Consequently, the first feasible $\Phi_{shared}$ encountered is the value that minimizes the distance to $UB_i$, satisfying the supremum condition for the freshness budget.

\begin{theorem}
Given a producer $\tau_i$ and consumer $\tau_j$, the offset $\Phi_i$ synthesized by the consensus algorithm minimizes the age of the consumed result: $Age_{optimized} \leq Age_{default}$.
\end{theorem}

\textit{Proof:}
\begin{itemize}
    \item \textbf{Age Definition:} $Age_j(\tau_{i,n(k)}) = f_{j,k} - f_{i,n(k)}$. Minimizing age requires maximizing the producer's completion time $f_{i,n(k)}$.
    \item \textbf{Constraint:} $f_{i,n(k)}$ must reside in $\mathcal{I} = [LB_i, UB_i]$, where $UB_i = s_{j,k} - (L_{ij} + C_j)$.
    \item \textbf{Optimization:} Our algorithm selects $\Phi_i$ to set $f_{i,n(k)} = \sup \{ t \in \mathcal{I} \mid t \leq UB_i \}$.
    \item \textbf{Conclusion:} For any baseline $\Phi_{default}$, $f_{baseline} \leq UB_i$. Since the optimized $f_{optimized}$ is the supremum, $f_{baseline} \leq f_{optimized} \leq UB_i$. Thus, $f_{j,k} - f_{optimized} \leq f_{j,k} - f_{baseline}$, yielding $Age_{optimized} \leq Age_{default}$. \hfill $\square$
\end{itemize}

\section{Conclusion}\label{sec:conclusion}

The evolution of real-time systems has shifted from ensuring task completion within strict deadlines to guaranteeing the data freshness of data across complex, multi-rate chains. This paper addressed the inherent "freshness-for-determinism" trade-off found in current automotive standards. While the Logical Execution Time (LET) paradigm provides structural determinism, it does so by injecting artificial latencies through output buffering. Conversely, standard Least Laxity First scheduling, while computationally efficient, suffers from data aging in multi-rate sensor fusion where fast data waits for high-latency dependencies.

We have proposed Scheduling based on \textbf{Data Freshness}, an inversion of the traditional deadline priority of real-time systems. By employing offset assignment based on data freshness, our constructive methodology assigns earlier static offsets to high-latency tasks, effectively "anchoring" the timeline to accommodate slower tasks in a task dependency chain. By delaying the critical tasks rather than buffering their outputs, we achieve \textbf{Just-In-Time (JIT)} data production, which aligns task execution with their result consumption at the succeeding tasks in multi-rate task-chains, minimizing data age. To resolve temporal conflicts in shared task structures, we introduced a \textbf{Consensus Search} backtracking algorithm that ensures global data freshness across heterogeneous consumer paths. Furthermore, we formally proved that this offset-based alignment \textbf{preserves the 100\% capacity of Global EDF}, demonstrating that data freshness can be guaranteed without sacrificing CPU utilization bounds.

Our work differentiates from the current state-of-the-art by moving from \textit{analysis} to \textit{construction}. While recent literature \cite{Gohari2022,Guenzel2023,Qingqiang2023,melani2015} provides rigorous methods to bound the Worst-Case Data Age (WCDA) and Worst-Case Response Time (WCRT) of existing schedules, our framework actively injects offsets to eliminate inherent data aging. Furthermore, unlike LET frameworks \cite{bini2023zero,bini2025jitter} that achieve predictability by delaying the availability of already-computed data, our method delays the sampling instant. This preserves the data freshness by ensuring that the multiple dependencies are synchronized to produce data as close as possible to their consumption, at the moment the slowest dependency becomes available.

%\section*{Acknowledgements}
%This work was partially funded by Fundação de Apoio da UFMG (Fundep), through Linha VI – Conectividade Veicular, a prioritary program from Mover (Mobilidade Verde e Inovação), project AutoDL (29271.03.01/2023.04-00).

\bibliographystyle{IEEEtran}
\bibliography{lisha,all,zes,let-rt}

\end{document}